\documentclass[conference]{IEEEtran}
\usepackage{amsmath, array}
\usepackage{amsfonts}
\usepackage{fancyhdr}
\usepackage[cyr]{aeguill}
\usepackage{cite}
\usepackage{epsfig,graphics,graphicx,latexsym}
\usepackage{amsmath,array,amscd,amsthm,latexsym,amssymb,mathrsfs,syntonly,eucal}
\usepackage{cases, longtable}
\usepackage[usenames]{color}
\usepackage{url}

\newtheorem{theorem}{Theorem}

\newtheorem{discussion}[theorem]{Discussion}

\begin{document}

\title{Hash-and-Forward Relaying for Two-Way Relay Channel}

\author{
\authorblockN{Erhan Y\i{}lmaz and Raymond Knopp}
\authorblockA{EURECOM, Sophia Antipolis, France\\
\{yilmaz, knopp\}@eurecom.fr }  }

\maketitle

\begin{abstract}

This paper considers a communication network comprised of two nodes, which have no mutual direct communication links, communicating two-way with the aid of a common relay node (RN), also known as separated two-way relay (TWR) channel.

We first recall a cut-set outer bound for the set of rates in the context of this network topology assuming full-duplex transmission capabilities. Then, we derive a new achievable rate region based on hash-and-forward (HF) relaying where the RN does not attempt to decode but instead hashes its received signal, and show that under certain channel conditions it coincides with Shannon's inner-bound for the two-way channel \cite{shannon_61}. Moreover, for binary adder TWR channel with additive noise at the nodes and the RN we provide a detailed capacity achieving coding scheme based on structure codes.

\end{abstract}

\section{Introduction}
\subsection{Related Work on Two-Way Relaying \label{subsec:RelatedWork}}

Up to now, different transmission schemes are proposed for two-way relay (TWR) channels \cite{rankov_asilomar05, wu/chou_ciss05, knopp_izs06, larsson_vtc06, katti_isit07, kim/mitran_conf07, oechtering_ciss07, avestimehr_allerton08, gunduz_allerton08, schnurr_isit08}. However, the capacity region for general TWR channel remains open. The simplest transmission scheme for TWR channel consists of four phases where the two nodes transmit their messages to the relay node (RN) successively and then the RN decodes and forwards each mobile's message in the following two time slots. However, using ideas from network coding (NC) \cite{ahlswede_jnl00}, it is shown in \cite{wu/chou_ciss05} that the last two transmissions may be merged into a single transmission, resulting in three time slots and hence a pre-log factor of $2/3$ with respect to the sum-rate. The number of required time slots for the communication between the two nodes can be reduced even further to two time slots by allowing them to simultaneously access the RN \cite{popovski_icc06, oechtering_ciss07, zhang_mobihoc06, kim/mitran_conf07}.

In \cite{rankov_isit06, knopp_izs06, popovski_icc06} amplify-and-forward (AF), decode-and-forward (DF) and compress-and-forward (CF) based relaying strategies consisting only of two time slots are proposed for TWR channel. A NC scheme that requires three time slots is considered in \cite{wu/chou_ciss05} where in the first two time slots the mobiles send their messages to the RN in orthogonal time slots, the RN decodes both messages and then combines them by means of the bit-wise XOR operation and retransmits it to the mobiles. There the mobiles are assumed to use the bit-wise XOR operation on the decoded message and the own transmitted message to obtain the message sent from the other mobile. Requiring three time slots, this bit-wise XOR based TWR scheme provides a pre-log factor of $2/3$ with respect to the sum-rate. In \cite{katti_sigcomm07} an analog network coding (ANC) scheme, where the RN amplifies and forwards mixed received signals, is proposed and compared to the traditional and digital network coding (bit-wise XOR at the RN) schemes in terms of network throughput. In \cite{popovski_icc06} the denoise-and-forward (DNF) relaying is proposed for TWR channel where the RN removes the noise from the combined mobiles' messages (on the multiple-access channel) before broadcasting and compared with AF, DF based TWR schemes as well as the traditional four phase scheme.

\begin{figure}[t]
\centering
\includegraphics[width=3.5in, height=2.2in, viewport = 10 90 700 490, clip]{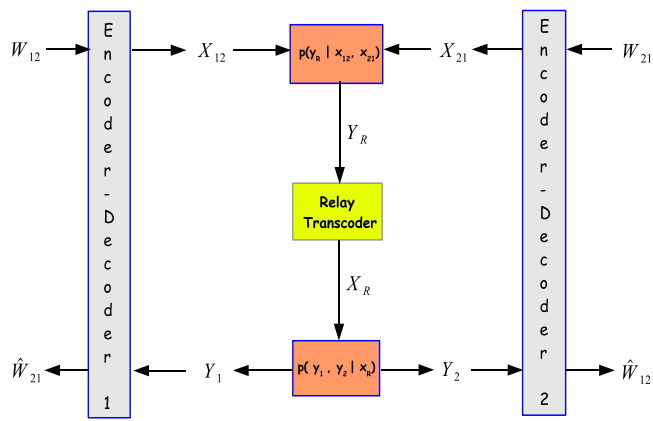}
\caption{Two-user two-way relay network.} \label{c6f:twoWayRelayEncoderDecoder}
\end{figure}

\subsection{Contributions \label{subsec:contributions}}

In this paper, we consider a TWR channel where two separated nodes seek to communicate, and can do so, via a common RN. This kind of scenario can occur in satellite communications where different ground stations want to exchange information or in public safety networks where different intervening entities (e.g. fire-fighters) want to communicate with each other to gain information about the current status at different parts of the disaster area \cite{knopp_meeting07}.

Throughout the paper we focus on discrete memoryless TWR channel model assuming full-duplex transmission capabilities at the communicating nodes. A cut-set outer bound for the set of rates in the context of this network topology is recalled and an achievable rate region based on DF relaying is provided where a message-level network coding enabling codebook size reduction at the RN is used. Then, a novel achievable rate region based on random binning, i.e. hash-and-forward (HF) relaying, is provided where the RN does not attempt to fully (or partially) decode but rather hashes its received signal so as to reduce the amount of information needed to send in the DL. Finally, we study a binary adder channel with additive noise at the nodes and the RN and show that the outer bound is virtually achievable by \emph{physical-layer} network coding, using group codes and partial decoding at the RN.

\section{Information-Theoretic System Model}

Consider the discrete-memoryless network in Fig.~\ref{c6f:twoWayRelayEncoderDecoder} comprising a pair of nodes communicating two-way via a RN. Node $a$ generates $n$-dimensional codewords $X_{ab}^n$ based on an index $W_{ab}\in\left\{1,2,\cdots,2^{nR_{ab}}\right\}$ where $R_{ab}$ is the information rate in bits/dimension in the direction $ab$, for $a \in \{1, 2\}$ and $b \in \{1, 2\} \setminus a$.
Each letter of the transmitted codewords, $X_{ab,i},i=1,2,\cdots,n$ belongs to an alphabet $\mathcal{X}_{ab}$ and is chosen according to a deterministic encoding function $X_{ab,i}=f_{ab,i}\left(W_{ab},Y_{a}^{i-1}\right)$ which includes the possibility for exploiting the past observations of the relay channel output (downlink). The RN also uses a deterministic transcoder which generates an $n$-dimensional output sequence $X_{R}^n$ based on the observed noisy multiuser (uplink) channel output $Y_{R}^n$. The transcoding function at the RN is causal and written as $X_{R,i}=f_{R,i}\left(Y_{R}^{i-1}\right)\in\mathcal{X}_R$, where $\mathcal{X}_R$ is the alphabet of the RN transmitter. The multiuser channel is memoryless and successive outputs are identically distributed according to the conditional probability $p(y_R|x_{12},x_{21})$. The observed sequences at the nodes, $Y_{1}^n$ and $Y_{2}^n$, are independent conditioned on the RN output and identically distributed according to the conditional probability $p(y_1,y_2|x_R)=p(y_1|x_R)p(y_2|x_R)$.

Receiver $a$ decodes its message based on the observed sequence $Y_{a}^n$ using a decoding function $\hat{W}_{ba}=g\left(Y_{a}^n,W_{ab}\right)$ with probability of error $P_{e,a}=\Pr\left(\hat{W}_{ba}\neq W_{ba}\right)$. 

\section{Outer Bound}\label{sec:outerbound}

Let us first derive an outer-bound on the capacity region from first principles. It is shown in \cite{yilmaz_thesis10} that by using Fano's inequality twice, we have the following outer-bound on achievable rates
\begin{equation}
R_{12} \leq \min\left\{\mathrm{I}\left(X_{12};Y_R|X_{21}\right),\mathrm{I}\left(X_R;Y_2\right)\right\}\label{eq:outerbound}.
\end{equation}

The first term (uplink) in \eqref{eq:outerbound} corresponds to the outer-bound of the two-way channel with a common output \cite{shannon_61} with which no coding strategy is known to coincide, except for additive channels and some special cases \cite{hekstra_jnl89}.  In our system, the achievable rates would be limited to this rate if the DL channels were capable enough to allow the RN to forward a sufficient characterization of $Y_R$ to the nodes.

\begin{discussion}
In the case of a noiseless DL channel our system boils down to a two-way channel with a common output \cite{shannon_61}. For a weak DL, the information rates will be limited by the second term in \eqref{eq:outerbound}. Again, were the uplink channel noiseless, an outer-bound to the achievable rates would be given by the second term.
\end{discussion}

\section{Forwarding Strategies at the Relay Node} \label{c6sec:DFHFCF}

At the node-1 and node-2 assuming coding schemes which do not exploit the DL signal in the encoding of the UL data so that the data streams are independent, we now consider coding schemes at the RN which do not attempt to decode or partially decode the transmitted data streams, but rather perform binning (hashing) of the received sequence. UL codebooks are generated randomly at each node according to $p(x_{12})$ and $p(x_{21})$. At the RN we will consider two different encoding strategies, namely DF and HF. Prior to discussing these further, we first look at DL coding for degraded broadcast channels (BCs) with successive refinement which will be needed for the probability of error analysis for data transmission on the DL.

\subsection{Degraded BC with successive refinement\label{sec:dbc_suc_refinement}}

We consider the case where the DL channel is degraded and a DL channel code where the two users are required to decode the same information but with different degrees of refinement (here two). By this we mean that the user with the weaker channel only receives the heavily-protected data stream, while the user with the stronger channel receives both. This is the general BC where there is common information for both users and extra information for the stronger user. We make use of standard multi-level coding \cite{book:Cover}, and let $R_{R,1}$ and $R_{R,2}$ denote the information rates of the RN to nodes 1 and 2 respectively, and assume further that $R_{R,1}\geq R_{R,2}$. The random codebook, $\{X_{R}(i,j), i = 1, 2, \cdots, 2^{n(R_{R,1} - R_{R,2})}, \;\;j = 1, 2, \cdots,2^{nR_{R,2}}\}$, is generated according to the distributions $p(u)$ and $p(x_{R}|u)$ in the standard-way so that the achievable DL rates are given by the convex hull (with respect to the parameter $\alpha$ and $p(x_{R}, u)$) of
\begin{align}
&(R_{R,1},R_{R,2}) = \displaystyle \big\{(R_{R,1},R_{R,2}): \notag \\
&\quad\quad\quad \operatornamewithlimits{\arg \max}_{{\footnotesize \begin{array}{c} p(x_{R}, u) \\ 0\leq\alpha\leq 1 \end{array}}} \{\alpha (R_{R,1} - R_{R,2})+(1-\alpha)R_{R,2} \} \big\}  \notag,
\end{align}
where
\begin{subequations}  \label{eq:dbc_suc_refinement_rates}
\begin{align}
R_{R,2}           &= \mathrm{I}(U; Y_2),         \label{eq:dbc_suc_refinement_rates1} \\
R_{R,1} - R_{R,2} &= \mathrm{I}(X_R; Y_1|U).     \label{eq:dbc_suc_refinement_rates2}
\end{align}
\end{subequations}

\vskip0.1in

\subsection{Decode-and-Forward Relaying} \label{ch6:subsec:DF}

In this section, we provide a brief view of the achievable rates for the digital single-relay network with two communicating nodes, wherein a generalized form of network-coding for discrete-memoryless channels is used.

The DL channel is a classical BC except for the fact that the decoders have side information to exploit, namely the transmitted codeword indices that they themselves sent during the UL portion. For simplicity, we assume that the two channel outputs are conditionally independent so that they can be separated into two transition probabilities $p(y_i|x_R), i = 1, 2$. Each node decodes the received sequence $y^n_1$ ($y^n_2$) using the side information from its own transmission to yield the estimates $\hat{\hat{W}}_{21}$ ($\hat{\hat{W}}_{12}$). In the following theorem, we provide an achievable rate region for the above channel model.

\begin{theorem}\label{th:DF_achievable_Rates}
An achievable rate region for the two-user single relay network using DF relaying strategy is given by the closure of the following set of inequalities
\begin{subequations} \label{eq:rate_df}
\begin{align}
R_{12}          &\;\leq \; \min \left\{\mathrm{I}(X_{12}; Y_R|X_{21}), \mathrm{I}(X_R; Y_2) \right\}, \label{eq:rate_df1} \\
R_{21}          &\;\leq \; \min \left\{\mathrm{I}(X_{21}; Y_R|X_{12}), \mathrm{I}(X_R; Y_1) \right\}, \label{eq:rate_df2} \\
R_{12} + R_{21} &\;\leq \; \mathrm{I}(X_{12}, X_{21} ; Y_R).                                          \label{eq:rate_df3}
\end{align}
\end{subequations}
\end{theorem}

\begin{proof}
The proof which is briefly explained in the following can be found in \cite{yilmaz_thesis10}.

The RN employs a multi-user receiver to decode the transmitted codeword indices yielding the estimates $(\hat{W}_{12}, \hat{W}_{21})$ at its output. without loss of generality, we assume $R_{12} \geq R_{21}$. Based on these indices, in the downlink phase, it then encodes the two indices using a two-dimensional codebook $\mathcal{X}_R = \{X^n_R(W_R^{(1)}, W_R^{(2)}), W_R^{(1)}=1, \ldots, 2^{n(R_{12}-R_{21})}, W_R^{(2)}=1, \ldots, 2^{n R_{21}}\}$. Here we note that the cardinality of the RN's codebook is at most $2^{n R_{12}}$.

To send indices $W_R^{(1)}$ and $W_R^{(2)}$, the RN chooses $X^n_R(W_R^{(1)}, W_R^{(2)})$ where $W_R^{(1)}=\left\lfloor \hat{W}_{12} 2^{-n R_{21}}\right\rfloor, \; W_R^{(2)}=\hat{W}_{12} \oplus \hat{W}_{21}$, where $x \oplus y$ denotes $(x+y) \; \mathrm{mod} \; 2^{n R_{21}}$.

At the receiver $1$ (the weak receiver), who has the side information $W_{12}$, the unique $\hat{\hat{W}}_{21}$ is chosen such that
\begin{align}
\{X^n_R(\left\lfloor W_{12} 2^{-n R_{21}}\right\rfloor, W_{12}~\oplus~\hat{\hat{W}}_{21}),Y^n_1 \}~\in~A^{n}_{\epsilon, 1}
\end{align}
where $A^{n}_{\epsilon, k}$ denotes the set of jointly-typical sequences $\{x_R(w_R^{(1)}, w_R^{(2)}), y_k\}$, $k=1,2$. Due to the side information $W_{12}$, $W_R^{(1)}$ is known to receiver $1$; hence the cardinality of the search space is limited to $2^{n R_{21}}$. If none or more than one exist an error is declared. The decoded index is then $W_{12} \oplus \hat{\hat{W}}_{21}$.

At the receiver $2$ (the stronger receiver), who has the side information $W_{21}$, the unique $\hat{\hat{W}}_{12}$ is chosen such that
\begin{align}
\{X^n_R(\left\lfloor \hat{\hat{W}}_{12} 2^{-n R_{21}}\right\rfloor, \hat{\hat{W}}_{12}~\oplus~W_{21}),Y^n_2 \}~\in~A^{n}_{\epsilon, 2}.
\end{align}
Note that the cardinality of the search space is $2^{n R_{12}}$. If none or more than one exist an error is declared. Due to the side information $W_{21}$, the decoded index at the receiver $2$ is given by $\hat{\hat{W}}_{12} 2^{-n R_{21}} + \hat{\hat{W}}_{12} \oplus W_{21}$.
\end{proof}

\subsection{Hash-and-Forward Relaying}

If the RN does not attempt to decode the received sequence $Y_R^n$, it may still attempt to reduce the amount of information needed to convey it to both users, and moreover with a higher degree of refinement for the user with the stronger DL channel. We will proceed similarly to distributed source coding as originally described by Slepian and Wolf in \cite{slepian_jnl73}. This was recently referred to as {\em hashing} at the RN, resulting in an {\em hash-and-forward (HF)} transmission scheme \cite{cover_isit07}.

\begin{theorem}\label{th:HF_achievable_Rates}
An achievable rate region for the discrete memoryless TWR network with HF relaying strategy is given by the closure of the following set of inequalities
\begin{subequations} \label{eq:hash_region}
\begin{align}
&R_{12} \leq \min\big\{\mathrm{I}(X_{12};Y_R|X_{21}), \notag\\
&\hskip1.2in \left[\mathrm{I}(U;Y_2)-H(Y_R|X_1,X_2)\right]^+ \big\} \label{eq:rate_hf1} \\
&R_{21} \leq \min \big\{\mathrm{I}(X_{21};Y_R|X_{12}), \notag\\
&\hskip0.3in \left[\mathrm{I}(X_R;Y_1|U) + \mathrm{I}(U;Y_2)-H(Y_R|X_1,X_2)\right]^+\big\} \label{eq:rate_hf2}
\end{align}
\end{subequations}
for UL joint probability $p(x_{12})p(x_{21})p(y_R|x_{12},x_{21})$ and the degraded BC with successive refinement defined in Section-\ref{sec:dbc_suc_refinement}.
\end{theorem}

\begin{proof}
The set of $\epsilon-$typical sequences $Y_R^n$ is partitioned into $2^{nR_{R,1}}$ bins $B_{ij},i=1,\ldots,2^{n(R_{R,1}-R_{R,2})}$, $j=1,\ldots,2^{n R_{R,2}}$, such that $\bigcup_{i,j} B_{ij}=A_\epsilon^n$ and $B_{ij}\bigcap B_{i'j'}=\emptyset,\mathrm{for}\;i\neq i' \mathrm{or}\;j\neq j'$. Let $\mathcal{B}(y_R^n)=(i,j)$ be a random hash-function which assigns a received sequence to a particular bin.

Encoding at the RN is done in two steps. First, if $y_R^n\in A_\epsilon^n$, we hash $y_R^n$ and let $W_{R,1}=i$ and $W_{R,2}=j$. If $y_R^n \notin A_\epsilon^n$, we set $(W_{R,1},W_{R,2})=(\mathrm{e},\mathrm{e})$ to indicate an error condition at the RN encoder. Then, we generate the transmitted sequence $X_R^n$ according to the multilevel-coding strategy for the (degraded) BC.

Decoder 1 creates a list, $L_1(y_R^n)$ of candidate $y_R^n$ based on the decoded $(\hat{i},\hat{j})$. The number of candidates in the list, $N_1(y_R^n)$ is bounded by
\begin{equation}
2^{n(H(Y_R)-R_{R,1}-\epsilon)}\leq N_1(y_R^n)\leq2^{n(H(Y_R)-R_{R,1}+\epsilon)}.
\end{equation}
Similarly, decoder 2 has a list based solely on $\hat{j}$, $L_2(y_R^n)$, for which the number of elements is bounded by
\begin{equation}
2^{n(H(Y_R)-R_{R,2}-\epsilon)}\leq N_2(y_R^n)\leq2^{n(H(Y_R)-R_{R,2}+\epsilon)}.
\end{equation}
Knowing $x_{12}^n$, decoder 1 tries to find an $x_{21}^n$ such that $(x_{12}^n,x_{21}^n, y^n)\in A_\epsilon^n$ for at least one $y^n\in L_1(y_R^n)$. If more than one $x_{21}^n$ or none are jointly $\epsilon$-typical, then an error is declared. Decoder 2 proceeds similarly and tries to find an $x_{12}^n$ knowing $x_{21}^n$ such that $(x_{12}^n, x_{21}^n, y^n)\in A_\epsilon^n$ for at least one $y^n \in L_2(y_R^n)$.

Assuming $(W_{12},W_{21}) = (1,1)$, the probability of decoding error (conditioned on receiving $(i,j)$ without error and $i\neq\mathrm{e},j\neq\mathrm{e}$) for decoder 1 is given by
\begin{align}
P_e^{(1)} &\leq \Pr\{(x_{12}^n(1),x_{21}^n(1),y_R^n)\notin A_\epsilon^n\} \notag\\
&\quad + 2^{nR_{21}} \left[\Pr\{(x_{12}^n(1),x_{21}^n(i\neq 1),y_R^n)\in A_\epsilon^n\} \right. \notag \\
&\quad \left. + N_1(y_R^n) \Pr\{(x_{12}^n(1),x_{21}^n(i\neq 1),{y'_{R}}^n) \in A_\epsilon^n \}\right] \notag\\
&\quad +\Pr\{(\hat{W}_{R,1},\hat{W}_{R,2})\neq(i,j)\} + \Pr\{W_{R,1}=\mathrm{e}\}. \label{eq:decoder1_hash_pe}
\end{align}
The first element in the sum can be made arbitrarily small by increasing $n$. The probability in the second term is the probability over the random ensemble of codebooks that an $x_{21}^n(i), i\neq 1$ is jointly $\epsilon-$typical with the true RN output (which is always in $L_1(y_R^n)$ if $(i,j)$ are received without error at the RN) and is given by
\begin{align}
&\Pr\left\{(x_{12}^n(1),x_{21}^n(i),y_R^n)\in A_\epsilon^n\right\} \notag\\
&\quad\quad\quad = \sum_{\footnotesize (x_{12}^n(1),x_{21}^n(i),y_R^n)\in A_\epsilon^n} p(x_{12}^n(1),y_R^n)) p(x_{21}^n(i))\notag\\
&\quad\quad\quad \leq 2^{n(H(X_{12},X_{21},Y_R)-H(X_{12},Y_R)-H(X_{21})-3\epsilon)} \notag\\
&\quad\quad\quad \leq 2^{-n(\mathrm{I}(X_{21};Y_R|X_{12})+3\epsilon)}.
\end{align}
The probability in the third term reflects the event that another ${y'_{R}}^n \in L_1(y_R^n)$ sequence from the random list of candidates is jointly $\epsilon-$typical with $x_{12}(1)$ and $x_{21}(i)$.  Note that this sequence is independent of both code sequences and thus
\begin{align}
&\Pr\left\{(x_{12}^n(1),x_{21}^n(i),{y'_R}^n)\in A_\epsilon^n\right\} \notag\\
&\quad\quad\quad = \sum_{\small(x_{12}^n(1),x_{21}^n(i),{y_R^{'}}^{n}) \in A_\epsilon^n} p(x_{12}^n(1)) p({y'_R}^n)) p(x_{21}^n(i))\notag\\
&\quad\quad\quad \leq 2^{n(H(X_{12},X_{21},Y_R)-H(X_{12})-H(Y_R)-H(X_{21})-4\epsilon)} \notag\\
&\quad\quad\quad \leq 2^{-n(\mathrm{I}(X_{21},X_{12};Y_R)-4\epsilon)}.
\end{align}
Combining the two probabilities and the size of the list in \eqref{eq:decoder1_hash_pe} yields
\begin{align}
P_e^{(1)} &\leq \Pr\{(\hat{i},\hat{j}) \neq(i,j)\} + 2 \epsilon + 2^{n(R_{21}-\mathrm{I}(X_{21};Y_R | X_{12}) - 3 \epsilon)} \notag \\
&\quad\quad + 2^{n(R_{21}-R_{R,1} + H(Y_R|X_{12},X_{21})-5\epsilon)}.
\end{align}
Proceeding in an identical fashion for decoder 2 yields
\begin{align}
P_e^{(2)} &\leq \Pr\{\hat{j}\neq j\} + 2 \epsilon + 2^{n(R_{12}-\mathrm{I}(X_{12};Y_R|X_{21}) - 3 \epsilon)} \notag \\
          &\quad\quad + 2^{n(R_{12}-R_{R,2} + H(Y_R | X_{12},X_{21}) - 5 \epsilon)}.
\end{align}
We now turn to the remaining error event in the overall error probability, namely the event that the bin indices $(i,j)$ for decoder 1 and $j$ for decoder 2 are incorrectly decoded. For the degraded BC with successive refinement (see Section-\ref{sec:dbc_suc_refinement}) we have the two error probabilities vanish if \eqref{eq:dbc_suc_refinement_rates} is satisfied.
\end{proof}

\begin{discussion}
Consider the special case where $p(y_1|x_R) = p(y_2|x_R)$, and thus $R_{R,1} = R_{R,2} = \mathrm{I}(X_R;Y_1)=\mathrm{I}(X_R;Y_2)$. Assume further the DL channels are very strong in the sense that
\begin{align*}
\mathrm{I}(X_R;Y_1) &\geq \; \max\left\{\mathrm{H}(Y_R|X_{12}), \mathrm{H}(Y_R|X_{21})\right\}.
\end{align*}
Under these conditions, the overall rate region coincides with Shannon's inner-bound for the two-way channel\cite{shannon_61} with a common output (which is not tight
due to the statistical independence of the input sequences) and is achievable by binning at the RN in the sense of Slepian-Wolf \cite{slepian_jnl73} with the only difference being that the {\em two} destinations now recover $Y_R$ without error using their own DL side information from the UL sequences.
\end{discussion}

\begin{figure}
\centering
\includegraphics[width=\linewidth, height=2.2in]{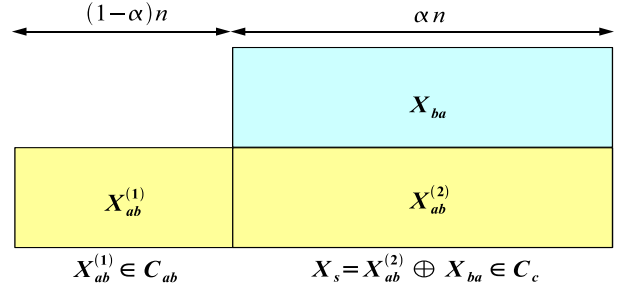}
\caption{Binary Adder Channel: Codebook Time-Sharing. \label{c6f:twoWayRelayTimeSharing}}
\end{figure}

\subsection{The capacity region of Binary Adder Channel}

We now restrict our treatment to a specific additive channel models, namely the binary adder channel. For this channel model, we will provide capacity achieving coding strategy based on linear group codes.

The input/output relationships for the binary adder channel are given by
\begin{align}
y_{R,n} &= x_{12,n} \oplus x_{21,n}\oplus z_{R,n}, \notag\\
y_{i,n} &= x_{R,n}  \oplus z_{i,n}, \quad i=1,2
\end{align}
where all variables are binary and $x \oplus y$ denotes the modulo-2 sum of $x$ and $y$. The probability of taking on the value 1 for the noise terms are denoted $\epsilon_{R}, \epsilon_{1},\epsilon_{2}$. Here the outer bound from Section-\ref{sec:outerbound} is
\begin{align}
R_{12} &\leq \min\left\{\mathrm{I}\left(X_{12};Y_R|X_{21}\right), \mathrm{I}\left(X_R;Y_2\right)\right\} \notag\\
       &= \min\big\{\mathrm{H}\left(Y_R|X_{21}\right)-\mathrm{H}\left(Y_R|X_{12},X_{21}\right), \notag\\
       &\hspace{1.3in} \mathrm{H}\left(Y_2\right)-\mathrm{H}\left(Y_2|X_R\right)\big\} \notag\\
       &\leq \min\left\{1-\mathcal{H}\left(\epsilon_{R}\right), 1-\mathcal{H}\left(\epsilon_{2}\right)\right\} \label{eq:bsc-outer}
\end{align}
where the last inequality becomes equality when the input distributions are uniform, i.e., $p(x_{12} = 1) = p(x_{ba} = 1) = 1/2$.

To show the achievability of \eqref{eq:bsc-outer} we assume that uplink and DL encoding occur in subsequent periods of $n$ output symbols, where $n$ is the length of code sequences. That is to say that after decoding a message from a group of $n$ symbols, the RN transmits the DL message while it receives the next UL message.  Suppose without loss of generality that $R_{12}\geq R_{21}$.

Consider two random codebooks $C_{12}$ with rate $R_{12}-R_{21}$ and $C_{c}$ with rate $R_{ba}$. $C_{c}$ is the {\em common codebook}. We now time-share between both codebooks as shown in Fig.~\ref{c6f:twoWayRelayTimeSharing} with $0 \leq \alpha\leq 1$. During the first time-slot of duration $\left(1-\alpha\right)n$ dimensions user $1$ transmits alone to the RN using $C_{12}$. Call the information sequence $X^{(1)}_{12}$. Using standard random coding arguments, arbitrarily small error probability for detection of $X^{(1)}_{12}$ is achievable if
\begin{align}
R_{12}-R_{21} &<    \left(1-\alpha\right)\mathrm{I}\left(X_{12};Y_R\right) \notag \\
              &\leq \left(1-\alpha\right)\left(1-\mathcal{H}\left(\epsilon_R\right)\right).
\end{align}

During the second time-slot of duration $\alpha n$ dimensions, both users transmit their information sequences $X^{(2)}_{12}$ and $X_{21}$ which are codewords belonging to the \emph{same linear code} over GF(2), $C_c$. As a result the RN receives the modulo-2 sum of the two codewords which is itself a codeword in $C_c$.
Linear codes achieve the capacity of the BSC (see \cite{book:Gallager}) and thus an arbitrarily small average error probability for the detection of $X^{(2)}_{12} \oplus X_{21}$ is possible if
\begin{align}
R_{21} \leq \alpha (1-\mathcal{H}(\epsilon_R)).
\end{align}

The RN encodes with a two-dimensional codebook of cardinality $2^{nR_{12}}$, indexed by column and row pair $(i,j), i=1,\cdots,2^{n(R_{12}-R_{21})},j=1,\cdots,2^{nR_{21}} $. Column $i$ is used to encode $X^{(1)}_{12}$ and row $j$ to encode $X^{(2)}_{12}\oplus X_{21}$. At receiver 1 (weak receiver), $X_{12}$ is known and so the column $i$ of the transmitted codeword is known.  Arbitrarily small error probability is achievable for detection of $j$ or $X^{(2)}_{12}\oplus X_{21}$ (and consequently $X_{21}$) if
\begin{align}
R_{21} < \mathrm{I}(X_R; Y_1) \leq 1-\mathcal{H}(\epsilon_1).
\end{align}
At receiver 2 (strong receiver) $X_{21}$ is known. Arbitrarily small error probability for detection of $(i,j)$ (or $X_{12}$) is achievable if
\begin{align}
R_{12} < \mathrm{I}(X_R; Y_2) \leq 1-\mathcal{H}(\epsilon_2).
\end{align}

Consider first the ``strong relay'' case where $\epsilon_{R}\geq\max\left(\epsilon_{1},\epsilon_{2}\right)$. Here we let $\alpha=1$ so that both users can achieve $1-\mathcal{H}(\epsilon_{R})$. For the ``medium relay'' case where $\epsilon_{1}\geq\epsilon_{R}\geq\epsilon_{2}$ we choose $\alpha=\frac{1-\mathcal{H}\left(\epsilon_{1}\right)}{1-\mathcal{H}\left(\epsilon_{R}\right)}$, so that $R_{12}=1-\mathcal{H}\left(\epsilon_{R}\right)$ and $R_{21}=1-\mathcal{H}\left(\epsilon_{1}\right)$ are achievable. Finally in the ``weak relay'' case when $\epsilon_{1}\geq\epsilon_{2}\geq\epsilon_{R}$ we choose $\alpha=\frac{1-\mathcal{H}\left(\epsilon_{1}\right)}{1-\mathcal{H}\left(\epsilon_{2}\right)}$, resulting in $R_{12}=1-\mathcal{H}\left(\epsilon_{2}\right)$ and $R_{21}=1-\mathcal{H}\left(\epsilon_{1}\right)$.


\section{Conclusions} \label{c6sec:conclusions}

In this paper, we studied different coding strategies for the discrete memoryless TWR channels. Specifically, a novel relaying strategy based on random binning of the receives signals at the RN, e.g. hash-and-forward relaying, was presented which coincides with Shannon's inner-bound for the two-way channel under certain channel conditions. For binary adder two-way relay channel, it is shown that the outer bound is achievable by physical-layer network coding, utilizing group codes and partial decoding at the RN. Our current work examines the use of feedback in the coding strategy at the terminals, an issue that was neglected here.

\section*{Acknowledgment}
This work was partially supported by the European Commission's 7th framework programme under grant agreement FP7-257616 also referred to as CONECT.

\bibliographystyle{IEEEtran}
\bibliography{../erhan_biblio}

\begin{thebibliography}{10}
\providecommand{\url}[1]{#1}
\csname url@samestyle\endcsname
\providecommand{\newblock}{\relax}
\providecommand{\bibinfo}[2]{#2}
\providecommand{\BIBentrySTDinterwordspacing}{\spaceskip=0pt\relax}
\providecommand{\BIBentryALTinterwordstretchfactor}{4}
\providecommand{\BIBentryALTinterwordspacing}{\spaceskip=\fontdimen2\font plus
\BIBentryALTinterwordstretchfactor\fontdimen3\font minus
  \fontdimen4\font\relax}
\providecommand{\BIBforeignlanguage}[2]{{%
\expandafter\ifx\csname l@#1\endcsname\relax
\typeout{** WARNING: IEEEtran.bst: No hyphenation pattern has been}%
\typeout{** loaded for the language `#1'. Using the pattern for}%
\typeout{** the default language instead.}%
\else
\language=\csname l@#1\endcsname
\fi
#2}}
\providecommand{\BIBdecl}{\relax}
\BIBdecl

\bibitem{shannon_61}
C.~E. Shannon, ``Two-way communication channels,'' in \emph{\em Proc. 4th
  Berkeley Symp. Math. Stat. Prob.,}, University of California Press, Berkeley,
  1961, pp. 611--644.

\bibitem{rankov_asilomar05}
B.~Rankov and A.~Wittneben, ``Spectral efficient signaling for half-duplex
  relay channels,'' in \emph{Proc.\ Conf. on Signals, Systems, and Computers},
  Pacific Grove, CA, Nov. 2005.

\bibitem{wu/chou_ciss05}
Y.~Wu, P.~A. Chou, and S.~Kung, ``Information exchange in wireless networks
  with network coding and physical-layer broadcast,'' in \emph{Proc.\ {IEEE}
  Conference of Information Sciences and Systems}, March 2005.

\bibitem{knopp_izs06}
R.~Knopp, ``Two-way radio network with a star topology,'' in \emph{\em Int.
  Zurich Seminar on Communications}, February 2006.

\bibitem{larsson_vtc06}
P.~Larsson, N.~Johansson, and K.~E. Sunell, ``Coded bi-directional relaying,''
  May 2006.

\bibitem{katti_isit07}
S.~Katti, I.~Maric, A.~Goldsmith, D.~Katabi, and M.~Medard, ``Joint relaying
  and network coding in wireless networks,'' in \emph{Proc.\ {IEEE} Int.\ Symp.
  on Information Theory}, June 2007.

\bibitem{kim/mitran_conf07}
S.~J. Kim, P.~Mitran, and V.~Tarokh, ``Performance bounds for bidirectional
  coded cooperation protocols,'' in \emph{in proc. of the 27th Int'l conf. on
  Distributed Computing Systems Workshops, ICDCSW}, June 2007.

\bibitem{oechtering_ciss07}
T.~J. Oechtering, C.~Schnurr, I.~Bjelakovic, and H.~Boche, ``Achievable rate
  region of a two phase bidirectional relay channel,'' in \emph{Proc.\ {IEEE}
  Conference of Information Sciences and Systems}, March 2007.

\bibitem{avestimehr_allerton08}
A.~S. Avestimehr, A.~Sezgin, and D.~N.~C. Tse, ``Approximate capacity of the
  two-way relay channel: A deterministic approach,'' in \emph{Proc.\ Conf. on
  Communication, Control, and Computing}, Illinois, USA, September 2008.

\bibitem{gunduz_allerton08}
D.~G\"{u}nd\"{u}z, E.~Tuncel, and J.~Nayak, ``Rate regions for the separated
  two-way relay channel,'' in \emph{Proc.\ Conf. on Communication, Control, and
  Computing}, Monticello, IL, September 2008.

\bibitem{schnurr_isit08}
C.~Schnurr, S.~Stanczak, and T.~J. Oechtering, ``Coding theorems for the
  restricted half-duplex two-way relay channel with joint decoding,'' in
  \emph{Proc.\ {IEEE} Int.\ Symp. on Information Theory}, Canada, July 2008.

\bibitem{ahlswede_jnl00}
R.~Ahlswede, N.~Cai, S.~Y.~R. Li, and R.~W. Yeung, ``Network information
  flow,'' \emph{{IEEE} Transactions on Information Theory}, vol.~46, pp.
  1204--1216, 2000.

\bibitem{popovski_icc06}
P.~Popovski and H.~Yomo, ``The anti-packets can increase the achievable
  throughput of a wireless multi-hop network,'' in \emph{Proc.\ {IEEE} Int.
  Conf. on Communication}, Istanbul, Turkey, June 2006.

\bibitem{zhang_mobihoc06}
S.~Zhang, S.~C. Liew, and P.~P. Lam, ``Physical-layer network coding,'' in
  \emph{\em in Proc. Annual International Conference on Mobile computing and
  Networking ({MOBIHOC})}, 2006.

\bibitem{rankov_isit06}
B.~Rankov and A.~Wittneben, ``Achievable rate region for the two-way relay
  channel,'' in \emph{Proc.\ {IEEE} Int.\ Symp. on Information Theory}, July
  2006.

\bibitem{katti_sigcomm07}
S.~Katti, S.~Gollakota, and D.~Katabi, ``Embracing wireless interference:
  {A}nalog {N}etwork {C}oding,'' in \emph{\em ACM {SIGCOMM}}, 2007.

\bibitem{knopp_meeting07}
R.~Knopp, ``Two-way wireless communication via a relay station,'' in
  \emph{GDRISIS meeting}, March 2007.

\bibitem{yilmaz_thesis10}
E.~Y\i{}lmaz, ``Coding strategies for relay based networks,'' \emph{Ph.D.
  dissertation, Telecom ParisTech}, July 2010.

\bibitem{hekstra_jnl89}
A.~P. Hekstra and F.~M.~J. Willems, ``Dependence balance bounds for
  single-output two-way channels,'' \emph{{IEEE} Transactions on Information
  Theory}, vol.~35, pp. 44--53, 1989.

\bibitem{book:Cover}
T.~M. Cover and J.~A. Thomas, \emph{Elements of {I}nformation {T}heory}.\hskip
  1em plus 0.5em minus 0.4em\relax United States: John Wiley \& Sons, 1991.

\bibitem{slepian_jnl73}
D.~Slepian and J.~Wolf, ``Noiseless coding of correlated information sources,''
  \emph{{IEEE} Transactions on Information Theory}, vol. IT-19, pp. 471--480,
  1973.

\bibitem{cover_isit07}
T.~M. Cover and Y.~H. Kim, ``Capacity of a class of {D}eterministic {R}elay
  {C}hannels,'' in \emph{Proc.\ {IEEE} Int.\ Symp. on Information Theory},
  Nice, France, 24-29 June 2007.

\bibitem{book:Gallager}
R.~G. Gallager, \emph{Information theory and reliable communication}.\hskip 1em
  plus 0.5em minus 0.4em\relax John Wiley \& Sons, 1968.

\end{thebibliography}

\end{document}